\documentclass{article}

\title{Reconstructing Volatility: Pricing of Index Options under Rough
Volatility}

\author{
  Peter K. Friz, Thomas Wagenhofer
  \tmaffiliation{TU and WIAS Berlin, TU Berlin}
}
\usepackage[english]{babel}
\usepackage{amsmath,amssymb,graphicx}
\usepackage{xcolor}
\usepackage{hyperref}
\usepackage{todonotes}
\usepackage{mathtools}
\usepackage{amsthm}

\newcommand{\assign}{\coloneqq}
\newcommand{\backassign}{\eqqcolon}
\newcommand{\tmaffiliation}[1]{\\ #1}

\newcommand{\tmop}[1]{\ensuremath{\operatorname{#1}}}
\newcommand{\tmtextbf}[1]{\text{{\bfseries{#1}}}}

\theoremstyle{plain}
\newtheorem{theorem}{Theorem}

\newtheorem{proposition}[theorem]{Proposition}
\newtheorem{corollary}[theorem]{Corollary}

\newtheorem{definition}[theorem]{Definition}

\newtheorem{remark}[theorem]{Remark}
\newtheorem{example}[theorem]{Example}

\newcommand{\cH}{\mathcal{H}}

\newcommand{\cK}{\mathcal{K}}

\newcommand{\cS}{\mathcal{S}}

\newcommand{\cW}{\mathcal{W}}

\begin{document}

\maketitle

\begin{center}
{\it To the memory of Marco Avellaneda}
\end{center} 
\begin{abstract}
In \cite{avellaneda2002reconstruction,avellaneda2003applicaton} Avellaneda et al. pioneered the pricing and hedging of index options - products highly sensitive to implied volatility and correlation assumptions -
with large deviations methods, assuming local volatility dynamics for all components of the index. We here present an extension applicable to non-Markovian dynamics and in particular the case of rough volatility dynamics.
\end{abstract}


\vspace{0.3cm}


\tableofcontents

\section{Introduction}

Given $N$ assets, with (discounted, risk-neutral) martingale dynamics
\begin{equation*}
    d S_t^i / S_t^i = \sigma_i (t, \omega) \,d B_t^i,
\end{equation*}
with Brownian driving noise, we consider an index of the form
\begin{equation*}
    I_t \assign \sum_{i = 1}^N w_i S^i_t
\end{equation*}
with the $w_i$'s constant.
From standard It\^o rules, assuming $\langle B^i, B^j \rangle_t / dt = \rho_{ij}$,
\begin{equation*}
\frac{d \langle I \rangle_t / dt}{I_t^2} = \sum_{i, j = 1}^N p^i_t p^j_t \rho_{i j} \sigma_i (t, \omega) \sigma_j (t, \omega) =:  \sigma^2_I (t,\omega), 
\end{equation*}
with 
$ p^i_t = w_i S^i_t / I_t$. Taking $t=0$, we have today's (deterministic) spot vols $\sigma_{i,0} :=\sigma_i(0,\omega)$, similar notation being used for the index. One then has the standard formula that relates spot volatilities of the index and its components:
\begin{equation} \label{equ:std}
        \sigma^2_{I,0} = \sum_{i, j = 1}^N  \rho_{i j} \sigma_{i,0}  \sigma_{j,0} p^i_0 p^j_0.
\end{equation}
In \cite{avellaneda2002reconstruction, avellaneda2003applicaton} 
we asked the question how to integrate volatility skew information more explicitly into \eqref{equ:std} and proposed a method for relating the implied volatility skew of the index to the implied volatility skew of the
components. Practical motivation, much related to Marco's activity at 
the time as head of the options research team at Gargoyle Strategic Investments, comes from dispersion trading: the strategy of selling (buying) index options, while buying (selling) options on the index components. 
The topic stayed close to Marco's heart and dispersion trading remained a topic in his NYU classes for years to come. 

The basic idea of these works was the use of short-time large deviations from diffusion processes, pioneered by \cite{varadhan1967diffusion}. This topic also stayed close to the heart of the first author of this note, as witnessed by \cite{de2013rational, deuschel2014marginal, deuschel2014marginal2, friz2015large, bayer2015pdf, de2018local, friz2021step, friz2021precise}.

\pagebreak[1]
The starting point of \cite{avellaneda2002reconstruction, avellaneda2003applicaton} is the familiar relation\footnote{In what follows we assume basic familiarity with stochastic, local and implied volatility, as found e.g. in  \cite{gatheral2006volatility,bergomi2016stochastic}. Formula \eqref{equ:IG} goes back to \cite{gyongy1986mimicking}, \cite{dupire1994pricing, derman1994riding}, also revisited in \cite{brunick2013mimicking, Cont2015forward}.}
   \begin{equation}  \label{equ:IG}
         \sigma^2_{I, \tmop{loc}} (t, I_t)=  \mathbb{E} \Bigl[  \sigma^2_{I, \tmop{stoch}}(t, \omega) \Big|  I_t  \Bigr]
   \end{equation}
together with assumed local volatility dynamics for the components of the index, that is 
\begin{equation*}
    \sigma_i (t, \omega) =\sigma_i (t, S^i),\quad i = 1,...,N.
\end{equation*}
 Setting $\tilde \sigma_i(x)=\sigma_i\bigl(0,S^i_0 e^{x_i}\bigr)$ and
also $\tilde \sigma_{I, \tmop{loc}} (\bar x) =\sigma_{I, \tmop{loc}} (0, I_0 e^{\bar x})$ 
it holds in the short-time limit that
\begin{gather*}
     \tilde \sigma^2_{I, \tmop{loc}} (\bar x) =\sum_{i,j=1}^N \rho_{ij} \tilde\sigma_i({x_i^*}) \tilde\sigma_j({x_j^*}) p_i(\mathbf{x}^*) p_j(\mathbf{x}^*).
\end{gather*}
where $\mathbf{x}^* \in \Gamma_{\bar x} = \{ \mathbf{x}: \sum_{i=1}^N w_i S^i_0 e^{x_i} = e^{\bar x} \}$ minimizes the distance to the origin  $\mathbf{x} = 0$, relative to the associated Riemannian metric \cite{varadhan1967diffusion}.
It is generically true (and here assumed - but see e.g. \cite{bayer2015pdf}) that $\mathbf{x}^*$ is unique. We also set 
\begin{gather*}
    p_i(\mathbf{x})=\frac{w_i S^i_0e^{x_i}}{\sum_{j=1}^N w_j S^j_0 e^{x_j}},
\end{gather*}
which represents the percentage of the stock $i$ in the index with $S^i =S^i_0e^{x_i}$. Furthermore $\mathbf{x}^*$ solves the non-linear system
\begin{align*}
    \int_0^{x_i^*} \frac{du}{\tilde \sigma_i(u)}&= \lambda \sum_{j=1}^N \rho_{ij} p_j(\mathbf{x}^*)\tilde \sigma_j(x_j^*), \qquad \forall i=1,\dots,N \\
    e^{\bar x}&=\sum_{i=1}^N w_i S^i(0) e^{x_i^*}.
\end{align*}
Here $\lambda$ corresponds to the Langrange multiplier of the price constraint $\Gamma_{\bar x}$.

Using approximate relations between local and implied volatility, notably the $1/2$-rule, valid in the short-time and ATM regime (see also \cite{gatheral2006volatility, gatheral2012asymptotics} and references therein),
this led us to
\begin{align*}
    \tilde \sigma_{I, \tmop{loc}}& (\bar x) = \\&
    \sqrt{\sum_{i,j=1}^N \rho_{ij}p_i(\mathbf{x}^*) p_j(\mathbf{x}^*) \Bigl(2\sigma_i^{\tmop{impl}}({x_i^*})-\sigma_i^{\tmop{impl}}(0)\Bigr) \Bigl(2\sigma_j^{\tmop{impl}}({x_j^*})-\sigma_j^{\tmop{impl}}(0)\Bigr) };
\end{align*}
together with a first order approximation for the most likely configuration $x^*$ from \cite[Equ.(15)]{avellaneda2002reconstruction}
\begin{equation} \label{equ:mostlikely}
    x_i^* = \frac{\bar{x}}{\sigma_I^2 (0)} \sum_{j = 1}^N \rho_{i j}, 
   w_j \sigma_i (0) \sigma_j (0),  \quad i=1,.. , N. 
 \end{equation}
Keep in mind, that $\tilde \sigma_{I, \tmop{loc}}$ corresponds to the local volatility at time $t=0$, hence equality holds above, despite using approximations.
 
With the harmonic average formula that expresses $\sigma_I^{\tmop{impl}}({\bar x})$ in terms of  $\tilde \sigma_{I, \tmop{loc}} (\bar x)$, this essentially concludes the task of refining \eqref{equ:std} in a tractable way that allows to integrate volatility skew information.
 
These (index) results were revisited and extended by various authors, including \cite[Sec.7.2]{henry2008analysis} and \cite[Sec.12.9]{guyon2014nonlinear}, notably towards local correlation model. The purpose of the present note is to revisit \cite{avellaneda2002reconstruction} in a way
that makes it clear that one can do, in fact, without the Markovian structure that seemed rather crucial in \cite{avellaneda2002reconstruction, avellaneda2003applicaton, henry2008analysis, bayer2015pdf} and related works. 
While in a diffusion setting, short-time can be considered as special case of small noise, cf.  \cite{osajima2015general, deuschel2014marginal}),
this is not so in a rough volatility setting and we should emphasize that we work in a small noise setting here, rather than the short-time regime of \cite{forde2017asymptotics}.
Last not least, following \cite{avellaneda2002reconstruction,avellaneda2003applicaton}  and to illustrate our approach we have kept the constant correlation structure, leaving any extension to 
stochastic and local correlations to future work.

{\bf Acknowledgements}: 
The first author acknowledges support from the German science foundation (DFG) via the cluster of excellence MATH+, as PI of project AA4-2.  The second author was supported by  the German science foundation (DFG) via MATH+, as PhD student in the Berlin Mathematical School (BMS). 
We thank the referee and several colleagues for valuable feedback.

\section{Index Options under Rough Volatility}

\subsection{Rough Volatility Index Dynamics }

We consider the model case where components follow rough volatility dynamics. For this let $\bigl(W^1,\dots,W^N,\overline{W}^1,\dots,\overline{W}^M\bigr)$ be independent Brownian motions and consider for $i=1,\dots,N$ a model of the form
\begin{align}
    d S_t^i / S_t^i &= f_i (\widehat{W}_t^i) \,d B_t^i,     \label{equ:2factorRV}
    \\
     B^i &= c_i W^i
   + \bar{c}_i \overline{W}^i, \qquad c_i^2 + \bar{c}_i^2 = 1,    \nonumber
   \\
   \widehat{W}_t^i &=
   \int_0^t K^{H_i} (t, s)\, d W_s^i , \quad K^H(t,s) = C(H) |t-s|^{H-1/2}.     \nonumber
\end{align}
Here $\widehat{W}^i$ is a Riemann-Liouville fractional Brownian motion with Hurst parameter $H_i \in (0, 1 / 2]$. 
The constant $C(H)$ is usually chosen such that $\widehat{W}$ has unit variance.
Furthermore the $c_i$
quantify the correlation between $B^i$ and $W^i$, i.e.\ between the respective
driving factors of the underlying and the stochastic volatility process. In
general, one would need to specify the full correlation structure of
\begin{equation*}
    \left( B^1, \ldots, B^N ; W^1 {, \ldots, W^N}  \right) . 
\end{equation*}
To keep things simple, we assume $c_i = - 1$, which is not an unreasonable assumption at
all in equity. Here the sign of $c_i$ does not matter, as one could redefine $f_i$ accordingly. We are
led to a path-dependent one factor stochastic volatility model,
\begin{equation*}
    d S_t^i / S_t^i = f_i (\widehat{W}_{t}^i) \,d W_t^i,
\end{equation*}
somewhat similar in spirit to \cite{hobson1998models} and \cite{guyon2014nonlinear}. As before we set
\begin{equation*}
    d \langle W ^i, W^j \rangle_t / d t = d \langle B ^i, B^j \rangle_t / d t =
   \rho_{i j} .
\end{equation*}

\subsection{Small Noise LDP for the Index under Rough Volatility}

We introduce the small noise problem
\begin{gather}\label{Eq_Model1}
\begin{aligned}
    d S_t^{i, \varepsilon} / S_t^{i, \varepsilon} &= f_i\bigl(\varepsilon \widehat{W
   _t}^i\bigr) \,d (\varepsilon W_t^i)
   \\
   I_t^{\varepsilon} &\assign \sum_{i = 1}^N w_i
   S^{i, \varepsilon}_t
   \end{aligned}
\end{gather}
Assume (w.l.o.g.) that $S_0^{i, \varepsilon}{\equiv}1$ and $\sum_{i=1}^N w_i = 1$, so that
$I_0^{\varepsilon} = 1$. Introduce the index log-price process
\begin{equation*}
    J ^{\varepsilon} = \log (I^{\varepsilon}), \quad \text{such that } I ^{\varepsilon} = e^{J^{\varepsilon}} .
\end{equation*} 
Write $\cH^N$ for the absolutely continuous paths from $[0, 1] \rightarrow \mathbb{R}^N$, started at zero, with $L^2$-derivative. Writing $\langle .,. \rangle $ for the $L^2$-inner product of both $\mathbb{R}^N$ and $\mathbb{R}$ valued paths, we
have the usual Cameron-Martin inner product
\begin{equation*}
    \langle h, h \rangle_{\cH^N} = \langle \dot{h}, \dot{h} \rangle = \sum_{i = 1}^N
   \langle \dot{h^i}, \dot{h^i} \rangle .
\end{equation*} 
For invertible $\rho$, we also define
\begin{equation*}
    \langle h, h \rangle_{\rho} \coloneqq \langle \dot{h}, \rho^{- 1} \dot{h} \rangle
   = \sum_{i, j = 1}^N \langle \dot{h^i}, (\rho^{- 1})_{i j} \dot{h^j} \rangle. 
\end{equation*}
If $W$ denotes a $N$-dimensional Brownian motion with covariance matrix
$\rho$, then $\varepsilon W$, viewed as $C ([0, 1], \mathbb{R}^N)$-valued
random variable, satisfies a LDP with good rate function $\langle h, h
\rangle_{\rho}  / 2$ whenever $h \in \cH^N$, and $+ \infty$ otherwise. One can
also treat non-invertible $\rho$, at the price of working with degenerate inner
products, so that only a proper subspace $\cH_{\rho} \subset \cH$ has a finite
rate function. 

\begin{theorem}
  Assume $f_i$ is smooth for $i = 1, \ldots, N$. Assume $\rho$ is
  invertible. Then $J_1^{\varepsilon} = \log (I_1^{\varepsilon})$ satisfies a
  LDP with speed $\varepsilon^2$ and good rate function 
  \begin{equation*}
    \Lambda (x) = \inf_{h \in \cH^N} \left\{ \frac{1}{2} \langle h, h
    \rangle_{\rho} : \phi (h) = x \right\} = \frac{1}{2} \langle h^x, h^x \rangle_{\rho}, 
  \end{equation*}
called {\em energy function},  where
  \begin{equation*}
      \phi (h) = \log \biggl( \sum_{i = 1}^N w_i \exp \bigl(\phi_i (h^i)\bigr) \biggr),
     \qquad \phi_i (h^i) = \int_0^1 f_i (\hat{h}^i)\,d h^i . 
  \end{equation*}
  The infimum is attained at some not necessarily unique $h^x \in \cH^N$.
  \tmtextbf{}
\end{theorem}
\begin{proof} If we had $J_1^{\varepsilon} = \phi(\varepsilon W^1,..., \varepsilon W^N)$ for a continuous map $\phi$, this would be a plain consequence of 
Schilder's theorem (LDP for Brownian motion) and the contraction principle. There are many ways to show that the results still holds, reviewed in \cite{bayer2023rough}. 
A standard method is by means of the so-called extended contraction principle, as in \cite{forde2017asymptotics}, see also \cite{jacquier2022large, gulisashvili2022multivariate}.
An alternative and arguably quite elegant argument was put forward in \cite{bayer2020regularity}, see also \cite{fukasawa2022partial}. Namely
$J_1^{\varepsilon}$ {\em is} the continuous image of $(\varepsilon W^1,..., \varepsilon W^N)$ {\em plus} certain iterated (It\^o) integrals, in the spirit of rough paths. For this enhanced noise, Schilder type large deviations are known and the result follows.
\end{proof}

\begin{remark}
    Note that ${\cH}^N \ni h = (h^1, \ldots, h^N) \mapsto \phi (h) \in \mathbb{R}$ is
smooth, to the extend the $f_i$ permit, and maps $0 \in \cH^N$ to $\phi (0) =
0 \in \mathbb{R}$. The function $\hat{h}^i$ is given by $\hat{h}^i_t=\int_0^t K^{H_i}(t,s)\,dh^i_s$.
\end{remark}

\subsection{Expansion of Abstract Energy Function}

Motivated by numerous papers on large deviations for stochastic and rough volatility we make the following definition.

\begin{definition}\label{Def_Ckrsbl} (i) Call \emph{$C^k$-reasonable} any map $\phi:\cH^N\rightarrow \mathbb{R}$ which is weakly continuous and $C^k$ in Fr\'echet sense with $k \ge 0$, further to 
\begin{equation*}
    \phi (0) = 0 \in \mathbb{R}, \quad   D \phi (0) \ne 0 \in \cH^N.
\end{equation*}
(ii) Call $\Lambda:\mathbb{R} \to [0,\infty]$ a \emph{$C^k$-good energy function} if, for some  $C^k$-reasonable $\phi$,  it is a good rate function of the form
\begin{align*}
    \Lambda (x) = \inf_{h \in \cH^N} \left\{ \frac{1}{2} \langle h, h \rangle_{\rho} : \phi (h) = x \right\}.
\end{align*}

\end{definition} 
Note $\Lambda (0 ) = 0$, with infimum trivially attained by $h^0 = 0 \in \cH^N$.
\begin{proposition}  \label{prop:Lnice} Any $C^1$-good energy function is continuous and increasing (resp. decreasing) on $\mathbb{R}^+$ (resp. $\mathbb{R}^-$).
\end{proposition} 
\begin{proof}
    Follows from Proposition \ref{Prop_MonRF} and Proposition \ref{Prop_ContRF}.
\end{proof}
\begin{remark} 
 In the special case of the rough Bergomi model, monoticity of the rate functions is shown in \cite[Lemma 15]{Gulisash2017}, with proof attributed to C.\ Bayer. 
\end{remark} 

The following ``abstract'' theorem gives an expansion of the rate function $\Lambda(x)$ and only involves the Fr\'{e}chet derivatives
\begin{equation*}
   \phi_0' : = D \phi (0) \in \cH^N,\qquad \phi_0'' \assign D^2 \phi (0) \in (\cH^N
   \times \cH^N)^{\star}.
\end{equation*}
As usual $(\cH^N \times \cH^N)^{\star}$ denotes the topological dual space of $\cH^N \times \cH^N.$

\begin{theorem}
\label{Thm_LambExp}
Assume $\Lambda$ is a $C^3$-good energy function.
  \begin{equation*}
      \Lambda (x) = \frac{1}{2} \langle h^x, h^x \rangle_{\rho} 
     = \left( \frac{x^2}{2} - \frac{x^3}{2\sigma_0^4} \phi_0'' \bigl(\rho \phi_0', \rho \phi_0'\bigr) \right)\Big/\sigma_0^2 + o (x^3)
  \end{equation*}
  where
 \begin{equation*}
    \sigma_0^2 \assign {{\langle \rho \phi_0', \rho \phi_0' \rangle_{\rho}}}=\langle \phi_0', \rho \phi_0' \rangle_{\cH^N}.
 \end{equation*} 
\end{theorem}

\subsection{Consequences for implied volatility in the small noise regime}

The following result has nothing to do with indices, and only assumes that
the asset price process $I_t^{\varepsilon}$ comes with a parameter
$\varepsilon > 0$, so that $\log (I_1^{\varepsilon})$, the log-price at time
$1$, satisfies a LDP with a good rate function $\Lambda = \Lambda (x)$. $\Lambda$ is assumed to be continuous and such that $\Lambda (x) = 0$ iff $x = 0$; cf. Assumption (A1) in \cite{friz2021precise}.

The following theorem can be seen as variation of the ``BBF formula'', \cite{berestycki2004computing} (short time) and also appears  in \cite{osajima2015general} (small noise). Remark that in short-time asymptotics for diffusion models,
also discussed in \cite{deuschel2014marginal,deuschel2014marginal2}, the energy function $\Lambda$ has the interpretation as geodesic point-subspace distance. Continuity and monotonicity of $\Lambda$ is then clear. In absence of this structure, authors including \cite{forde2017asymptotics, gulisashvili2022multivariate} express their large deviations in terms of $\Lambda^\star(x) = \inf_{y>x} \Lambda(y)$ whenever $x\ge 0$ and $\Lambda^\star(x) = \inf_{y<x} \Lambda(y)$ for $x<0$. As will be pointed out in the proof, this is not necessary in our setting, 
despite dealing with a somewhat generic non-Markovian small noise situation. 

\begin{theorem}[Implied volatility]
\label{Thm_ImplVol}
Under the above assumption it follows that
\begin{equation*}
    (\sigma^{\varepsilon}_{\tmop{impl}})^2 (x, 1) \sim_{\varepsilon
     \downarrow 0} \frac{x^2}{2 \Lambda (x)} .
\end{equation*}
\end{theorem}

\begin{corollary}[Spot variance and skew] 
\label{Cor_SpotSkew}
Let $\Lambda$ have the local expansion of
  Theorem \ref{Thm_LambExp}. Then
  \begin{equation*}
      (\sigma^{\varepsilon}_{\tmop{impl}})^2 (x, 1) \sim_{\varepsilon
     \downarrow 0} \frac{x^2}{2 \Lambda (x)} \sim_{x \downarrow 0} \sigma_0^2
     + x\mathcal{S}_0 + o (x)
  \end{equation*} 
  with spot variance implied variance skew given by, respectively,
  \begin{align*}
      {\sigma^2_0} = & \langle \phi_0', \rho \phi_0' \rangle_{\cH^N},
      \\
      {\mathcal{S}_0} = & \phi_0'' (\rho   \phi_0', \rho   \phi_0') /
      \sigma_0^2 .
  \end{align*}
\end{corollary}

\begin{proof}[Proof of Theorem \ref{Thm_ImplVol}]
We content ourselves with a sketch of the argument. The LDP together with Proposition \ref{prop:Lnice} give directly OTM binary call price aymptotics,
  \begin{equation*}
    \mathbb{P} [X_1^{\varepsilon} > x] \approx \exp \Bigl\{ - \frac{\Lambda(x)}{\varepsilon^2} \Bigr\} .
  \end{equation*}
  Matching exponents with OTM Bachelier prices and their Gaussian tail
  behavior,
  \begin{equation*}
      \mathbb{P} [\sigma \varepsilon W_1 > x] \approx \exp \Bigl\{ -
     \frac{x^2}{2 \sigma^2 \varepsilon^2} \Bigr\}.
  \end{equation*}
  For the effective normal implied volatility, as $\varepsilon \downarrow
  0$, we see that
  \begin{equation*}
  (\sigma^{\varepsilon}_{\tmop{norm}})^2 (x, 1)  \sim \frac{x^2}{2 \Lambda(x)} .
  \end{equation*}  
  The same asymptotics is valid for $\sigma^{\varepsilon}_{\tmop{impl}}$, the
  Black-Scholes implied volatility.  This follows from the fact that large asymptotics
  for binary and classical (OTM) option are identical. The only caveat here is
  a $1 +$ moment assumption to treat call options, cf.\ \cite{friz2021precise}.
\end{proof}

\begin{proof}[Proof of Corollary \ref{Cor_SpotSkew}]
    We only need to give a simple expansion for small $x$. Due to Theorem \ref{Thm_LambExp}:
    \begin{align*}
    \frac{x^2}{2 \Lambda (x)}&=\sigma_0^2 \left( \frac{1}{1 - x \sigma_0^{-4} \phi_0'' (\rho \phi_0', \phi_0') + o (x)} \right) 
    \\
    &= \sigma_0^2 \Bigl(1 + x \sigma_0^{-4} \phi_0'' (\rho \phi_0', \phi_0') + o (x)\Bigr) .
    \end{align*}
\end{proof}

\subsection{Index Spot Variance and Skew}

We now re-introduce, step-by-step, the structure of interest to us. We start
with a general index result, that applies, for instance to the index with
local volatility components considered in \cite{avellaneda2002reconstruction}. In this case $\phi_i (h^i) =
y_1^{\varepsilon}$ is the time-$1$ solution map of the ODE 
\begin{equation*}
    d y^i = f_i (y^i)\, dh^i,
\end{equation*}
with initial value $y^i = 0$. For Lipschitz $f_i$ this solution map is well-posed. This result also applies to rough volatility, with
\begin{equation*}
    \phi_i (h^i) = \int_0^1 f_i (\hat{h}^i)\, d h^i .
\end{equation*}
Note that in this case $\langle \phi_{i0}', k^i \rangle_\cH = \int_0^1 f_i (0) d k^i$  implying that $\phi'_{i0} \equiv D \phi_i (0) \in \cH$ has, as element of the Cameron-Martin space, constant
velocity $\sigma_i = f_i (0)$. A similar statement holds in the local
volatility example. This motivates our condition \eqref{Eq_SpeedAss} below. The
following result is a consequence of the general Corollary \ref{Cor_SpotSkew}, injecting the
additional information of weights and correlations.

\begin{proposition}[Index energy]
    \label{Prop_IndEgy}
    Assume that $\sum_{i = 1}^N w_i = 1$ and let $\phi_i :  \cH  \rightarrow
  \mathbb{R}$ be $C^3$-reasonable\footnote{in the sense of Definition \ref{Def_Ckrsbl} with $N=1$. } with $i= 1, \ldots, N$. Let
  \begin{equation*}
      \exp\bigl(\phi (h^1, \ldots, h^N)\bigr) \coloneqq \sum_{i = 1}^N w_i \exp \bigl(\phi_i (h^i)\bigr).
  \end{equation*}
  Set $\phi_{i0}' = D \phi_i (0) \in \cH$ as well as $\phi_{i0}'' = D^2 \phi_i (0) \in (\cH \times \cH)^{\star}$ and assume 
  \begin{equation}\label{Eq_SpeedAss}
      \quad \phi_{i 0}' = \sigma_i \tmop{Id} \in \cH, \qquad \tmop{Id}: t \mapsto t.
  \end{equation}
  Then 
 \begin{equation*}
      (\sigma^{\varepsilon}_{\tmop{impl}})^2 (x, 1) \sim_{x \downarrow 0} \sigma_I^2
     + x\mathcal{S}_I + o (x),
  \end{equation*}
  with
  \begin{align}
  \nonumber
      \sigma_I^2 & =  \sum_{i, j=1}^N w_i w_j \rho_{i j} \sigma_i \sigma_j,\\
      \label{Eq_SkewConstSpeed}
       \mathcal{S}_I & =  - \sigma_I^2 + \biggl( \sum_{i=1}^N w_i \Sigma_i^2
       \bigl(\phi_{i 0}'' (\tmop{Id}, \tmop{Id}) + \sigma_i^2\bigr) \biggr) \Big/
       \sigma_I^2, \qquad \Sigma_i \assign \sum_{j=1}^N w_j \rho_{i j} \sigma_j. 
  \end{align}
 \end{proposition}

\begin{remark} (i) Assumption \eqref{Eq_SpeedAss} is satisfied in large classes of examples. However, as the proof shows one can do without but the formulae are a bit less pretty. 
(ii)  The expression for spot variance $\sigma_I^2$ is consistent with the classical
 formula for Index Options that we gave in \eqref{equ:std}, cf. \cite{avellaneda2002reconstruction} or \cite{guyon2014nonlinear}. The expression for $\mathcal{S}_I$ can be seen an answer to the problem, first
  tackled in \cite{avellaneda2002reconstruction}, of how to integrate volatility skew information into such a classical formula. 
\end{remark}

\begin{remark}[Most-likely configuration] 
Consider, up to first order, 
$\bar{x} = \phi (\bar{h}) \coloneqq $ $ \sum_{i = 1}^N w_i \phi_i(\bar{h}^i)$ 
where 
$\bar{h} = (\bar{h}^1, \ldots, \bar{h}^N)$ 
is the minimizer for the constraint 
$\bar{x} = \phi (\bar{h})$. 
We then know 
$\bar{h} = \bar{x} a$, as well as 
$\phi (\bar{h})\coloneqq \bar{x} \langle \phi'_0, a \rangle_{H^N}$
with 
\begin{equation*}
   a^i \assign \frac{(\rho   \phi_0')^i}{{\langle \rho \phi_0', \rho \phi_0'
   \rangle_{\rho}} } = \frac{(\rho   \phi_0')^i}{\sigma_I^2 (0)} \Rightarrow
   \bar{h}^i = \frac{\bar{x}}{\sigma_I^2 (0)} (\rho   \phi_0')^i \quad i = 1,
   \ldots, N. 
\end{equation*} 
We thus have that 
$\bar{x}^i = \phi_i (\bar{h}^i)$
 equals, to first order and in agreement with \eqref{equ:mostlikely},
\begin{align*}
    \langle \phi'_{i 0}, \bar{h}^i \rangle_{H } &= \frac{\bar{x}}{\sigma_I^2
   (0)} \langle \phi'_{i 0}, (\rho   \phi_0')^i \rangle_{H } =
   \frac{\bar{x}}{\sigma_I^2 (0)} \sum_{j = 1}^N \rho_{i j} w_j \langle
   \phi'_{i 0}, \phi_{j 0}' \rangle_{H } 
   \\
   &= \frac{\bar{x}}{\sigma_I^2 (0)}
   \sum_{j = 1}^N \rho_{i j} w_j \sigma_i (0) \sigma_j (0) ,
\end{align*} 
where we used 
$\phi_0' = \sum_{j = 1}^N w_j \phi_{j 0}' \Rightarrow (\rho  
\phi_0')^i = \sum_{j = 1}^N \rho_{i j} w_j \phi_{j 0}'$. 
\end{remark} 

\begin{proof}[Proof of Proposition \ref{Prop_IndEgy}]
  Let $k = (k^1, \ldots, k^N) \in \cH^N$. We apply the function $ \exp\bigl(\phi (h)\bigr)
  = \sum_{i = 1}^N w_i \exp\bigl(\phi_i (h^i)\bigr)$ with $h = \varepsilon k$. The
  l.h.s.\ of $\exp (\phi (\varepsilon k))$ then expands to
\begin{equation*}
    1 + \phi (\varepsilon k) + \phi^2  (\varepsilon k) / 2 + o
     (\varepsilon^2) = 1 + \varepsilon \langle \phi_0', k \rangle_{\cH^N} +
     \frac{\varepsilon^2}{2} \Bigl(\phi_0'' (k, k) + \langle \phi_0', k
     \rangle^2_{\cH^N}\Bigr) + o (\varepsilon^2).
\end{equation*}
The r.h.s.\ expands to the weighted sum of the same expression with $\phi (\varepsilon k)$ replaced by $\phi_i  (\varepsilon k^i)$, namely
\begin{equation*}
      1 + \varepsilon \sum_{i = 1}^N w_i \langle \phi_{i 0}', k^i \rangle_{\cH }
     + \frac{\varepsilon^2}{2} \sum_{i = 1}^N w_i \Bigl(\phi_{i 0}'' (k^i, k^i) +
     \langle \phi_{i 0}', k^i \rangle_{\cH }^2\Bigr) + o\bigl(\varepsilon^2\bigr).
\end{equation*}
Power matching gives
\begin{equation*}
      \langle \phi_0', k \rangle_{\cH^N} = \sum_{i = 1}^N w_i \langle \phi_{i
     0}', k^i \rangle_{\cH },
\end{equation*}
implying that $\phi_0' = \bigl(w_1 \phi_{10}', \ldots, w_N \phi_{N 0}'\bigr) \in \cH^N$.
For the second order
  \begin{equation*}
      \phi_0'' (k, k) = - \langle \phi_0', k \rangle^2_{\cH^N} + \sum_{i = 1}^N
     w_i \Bigl(\phi_{i 0}'' (k^i, k^i) + \langle \phi_{i 0}', k^i \rangle_{\cH }^2\Bigr).
  \end{equation*}
We note that
  \begin{equation*}
    \bigl(\rho \phi_0'\bigr)^i = \sum_{j=1}^N \rho_{i j} w_j \phi_{j 0}'.
  \end{equation*}
We enter this expression into $\mathcal{S}_I = \phi_0'' (\rho   \phi_0',
  \rho   \phi_0') / \sigma_I^2$ , see Corollary \ref{Cor_SpotSkew}. This gives
  \begin{equation*}
       \mathcal{S}_I = (\mathcal{S}_1 +\mathcal{S}_2 +\mathcal{S}_3) /
     \sigma_I^2,
  \end{equation*}
with
  \begin{gather*}
      \mathcal{S}_1 = - \Biggl( \sum_{i, j = 1}^N w_i w_j \rho_{i j} \langle
     \phi_{i 0}', \phi_{j 0}' \rangle_{\cH} \Biggr)^2 = - \sigma_I^4
  \end{gather*}
and
  \begin{align*}
      \mathcal{S}_2 &= \sum_{i = 1}^N w_i \phi_{i 0}'' (k^i, k^i)_{}  = \sum_{i,
     j, \ell=1}^N w_i w_j w_{\ell} \rho_{i j} \rho_{i \ell} \phi_{i 0}'' (\phi_{j
     0}', \phi_{\ell 0}')_{} 
     \\
     &= \sum_{i=1}^N w_i \phi_{i 0}'' (\tmop{Id}, \tmop{Id})
     \biggl( \sum_{j=1}^N w_j \rho_{i j} \sigma_j \biggr)^2 .
  \end{align*}
where the last equality holds under the assumption of \eqref{Eq_SpeedAss}. Here $\tmop{Id}$ denotes the scalar Cameron--Martin path $t \mapsto t$ with velocity $1$. At last,
  \begin{equation*}
      \mathcal{S}_3 = \sum_{i = 1}^N w_i \biggl( \sum_{j=1}^N \rho_{i j} w_j \langle
     \phi_{i 0}', \phi_{j 0}' \rangle_{\cH } \biggr)^2 = \sum_{i = 1}^N w_i
     \sigma_i^2 \biggl( \sum_{j=1}^N w_j \rho_{i j} \sigma_j \biggr)^2,
  \end{equation*}
where the last equality again holds if we assume \eqref{Eq_SpeedAss}.

Set $\Sigma_i = \sum_{j=1}^N w_j \rho_{i j} \sigma_j $ and note, always under the assumption of \eqref{Eq_SpeedAss},
  \begin{equation*}
      \mathcal{S}_2 = \sum_{i=1}^N w_i \phi_{i 0}'' (\tmop{Id}, \tmop{Id})
     \Sigma_i^2, \qquad \mathcal{S}_3 = \sum_{i = 1}^N w_i \sigma_i^2
     \Sigma_i^2.
  \end{equation*}
In summary, we conclude by writing
  \begin{equation*}
      \mathcal{S}_I = \frac{\mathcal{S}_1 +\mathcal{S}_2
     +\mathcal{S}_3}{\sigma_I^2} = - \sigma_I^2 + \biggl( \sum_{i=1}^N w_i \Sigma_i^2
     \Bigl(\phi_{i 0}'' (\tmop{Id}, \tmop{Id}) + \sigma_i^2\Bigr) \biggr) \bigg/ \sigma_I^2 .
  \end{equation*}
 
\end{proof}
\begin{remark}
    Consider the case $N = 1$. In this case $w_1 = 1, \Sigma_1 = \sigma_1$ and $\sigma_I  = \sigma_1 $. Hence $\mathcal{S}_I $ reduces to $\phi_{10}'' (\tmop{Id}, \tmop{Id})$, in agreement with the skew expression of Proposition \ref{Cor_SpotSkew}, applied with $\rho = \pm 1$.
\end{remark}

\section{Return to rough volatility} 

We now return to the model defined by \eqref{Eq_Model1}, i.e.\ dynamics of the form
\begin{gather*}
\begin{aligned}
    d S_t^{i, \varepsilon} / S_t^{i, \varepsilon} &= f_i\bigl(\varepsilon \widehat{W
   _t}^i\bigr) \,d (\varepsilon W_t^i)
   \\
   I_t^{\varepsilon} &= \sum_{i = 1}^N w_i
   S^{i, \varepsilon}_t
   \end{aligned}
\end{gather*}
Given $H\in (0,1/2]$ we choose the kernel $K^H (t, s) = \sqrt{2 H} (t - s)^{H - 1 / 2}$, such that $\widehat{W}_1 = \int_0^1 K^H (1, s) d W_s$ has unit
variance, cf.\ \cite{bayer2019shorttime}. With this kernel we define the operator $\cK^H:\cH\rightarrow \cH$ such that $\bigl(\cK^H h\bigr)(t)=\int_0^t h(s)K^H(t,s)\,ds$. Note that a short calculation shows
\begin{equation*}
    \langle \cK^{H } 1, 1 \rangle = \frac{\sqrt{2 H}}{\left( H + \tfrac{1}{2}
   \right) \left( H + \tfrac{3}{2} \right)}.
\end{equation*}

\subsection{Single Asset, One-factor Rough Volatility Dynamics}
As as warm-up, we consider the case of $N = 1$ asset, with trivial correlation
``matrix'' $1$. By Corollary \ref{Cor_SpotSkew}, 
    \begin{align*}
        {\sigma_0} ^2 & =  \langle \phi_0', \phi_0' \rangle_{\cH },
        \\
        \mathcal{S}  & = \phi_0'' (  \phi_0',   \phi_0') / \sigma_0^2.
    \end{align*}
We can therefore find all relevant terms by expanding $\phi(\varepsilon k)$ to order $o (\varepsilon^2)$:
    \begin{align*}
        \phi (\varepsilon k) &= \int_0^1 f  (\varepsilon \hat{k} ) d (\varepsilon k) 
        \\
        &\approx  \varepsilon \biggl( \int_0^1 f  (0) d k \biggr) + \varepsilon^2
        \int_0^1 f'  (0) \hat{k} d k = \varepsilon f_0 \langle 1, \dot{k} \rangle +
        \frac{\varepsilon^2}{2} 2 f_0' \langle \cK^H \dot{k}, \dot{k}\rangle.
    \end{align*}
From this we read off $\langle \phi_0', k \rangle_\cH = f_0 \langle 1, \dot{k}
    \rangle$ as well as $\phi_0'' (k, k) = 2 f_0' \langle \cK^H \dot{k}, \dot{k} \rangle$. In
particular  $\phi_0'$ has constant velocity 
    \begin{equation*}
        \dot{\phi_0'} \equiv f_0.
    \end{equation*}
If $f\in C^2$ and $f_0\not =0$, then $\phi$ is also $C^k$-reasonable, see Definition \ref{Def_Ckrsbl}. By Corollary \ref{Cor_SpotSkew} we see
    \begin{align*}
     \sigma_0^2 & =  \langle \phi_0', \phi_0' \rangle_{\cH } = \langle
     \dot{\phi_0'}, \dot{\phi_0'} \rangle  = f_0^2,\\
     \mathcal{S}_0 & = \phi_0'' (  \phi_0',   \phi_0') / \sigma_I^2 = 2 f_0'
     \langle \cK^H 1, 1 \rangle.
    \end{align*}
By the chain-rule,
    \begin{equation*}
        \mathcal{S}_0 \assign \partial_x \bigl(\sigma^0_{\tmop{impl}}\bigr)^2 (x, 1) |_{x = 0} 
        =2f_0 \partial_x \sigma^0_{\tmop{impl}} (x, 1) |_{x = 0}  .
    \end{equation*}
   Hence in the small noise regime we have the following ATM implied volatility skew:
    \begin{align*}
        \partial_x \sigma^0_{\tmop{impl}} (x, 1) |_{x = 0}  =
     \frac{\mathcal{S}_0}{2 f_0} = \frac{f_0'}{f_0} \langle \cK^H 1, 1
     \rangle.
    \end{align*}
\begin{remark}
    This result is consistent with the skew formula in \cite{bayer2019shorttime}, see equation \eqref{Eq_Skew} below. 
    \begin{gather}
        \label{Eq_Skew}
        \frac{\sigma_{\mathrm{impl}}\left(y t^{1 / 2-H+\beta}, t\right)-\sigma_{\text {impl }}\left(z t^{1 / 2-H+\beta}, t\right)}{(y-z) t^{1 / 2-H+\beta}} \sim \rho \frac{\sigma_0^{\prime}}{\sigma_0}\langle \cK^H 1,1\rangle t^{H-1 / 2}.
    \end{gather}
    Since we deal with small noise rather than short time, there is no extra $t^{H - 1 / 2}$ blowup factor here!
\end{remark}

\subsection{Index with One-factor Rough Volatility Components }

Consider now $N \in \mathbb{N}$ assets, with (non-degenerate) correlation matrix
$\rho$. Using notation from Proposition \ref{Prop_IndEgy}, we can recycle the single asset computations. For $i = 1, \ldots , N$,
\begin{equation*}
    \sigma_{i0}  =  f_{i0}  \assign f_i  (0), \qquad \phi_{i0}'' = 2
     f_{i0}' \langle \cK^{H_i} 1, 1 \rangle.
\end{equation*}
The proof of Proposition \ref{Prop_IndEgy} shows that $\phi_0'\not = 0$ if $f_{i0}\not=0$ for some $i$, which together with $C^2$-regularity of $f$ implies that $\phi$ is $C^3$-reasonable, see Definition \ref{Def_Ckrsbl}.
Application of the second part of Proposition \ref{Prop_IndEgy} gives index spot variance and skew
\begin{gather}\label{Eq_SpotVarSkew}
\begin{aligned}
     \sigma_I^2 & =  \sum_{i, j=1}^N w_i w_j \rho_{i j} f_{i0}  f_{j0},
     \\
       \mathcal{S}_I & =  - \sigma_I^2 + \biggl( \sum_{i=1}^N w_i \Sigma_i^2
       \Bigl(2 f_{i0}' \langle \cK^{H_i} 1, 1 \rangle + f_{i0}^2\Bigr)\biggr)\Big/
       \sigma_I^2,  
       \end{aligned}
\end{gather}
where $\Sigma_i = \sum_{j=1}^N w_j \rho_{i j} f_{j0}$. 

\begin{example}[Index with One-factor rough Bergomi components]
For $i =1,\dots,N$ consider component dynamics of ``rough Bergomi'' type, following the terminology of \cite{bayer2016pricing}, i.e.
\begin{equation*}
     dS_t^i / S_t^i = \sigma_i e^{\eta_i \widehat{W}_{t}^i} \bigl(c_i \,dW^i + \bar{c}_i \,d\overline{W}^i\bigr),
\end{equation*}
with vvol $\eta_i > 0$ 
and spot volatility $\sigma_i$. 
We consider the ``fully correlated'' 
case, $c_i = - 1$, hence $\bar{c}_i = 0$. In law, this is the same
as
\begin{equation*}
    d S_t^i / S_t^i = \Bigl( \sigma_i e^{- \eta_i \widehat{W}_{t}^i} \Bigr)\,d
   W^i_t \equiv f_i \bigl(\widehat{W}_{t}^i\bigr)\,d W^i_t .
\end{equation*}
In the notation of this section we have $f_{i0}\coloneqq f_i  (0) = \sigma_i$ and  ${f'_{i0}}  = - \eta_i \sigma_i$. Thus the
spot volatility and skew are given by
\begin{equation*}
    \sigma_{i0}  = \sigma_i , \qquad \frac{\mathcal{S}_{i0}}{2 \sigma_i} = -
   \eta_i \langle K^{H_i} 1, 1 \rangle = - \eta_i \frac{\sqrt{2 H_i}}{\left(
   H_i + \tfrac{1}{2} \right) \left( H_i + \tfrac{3}{2} \right)} .
\end{equation*}
Concerning the index $I = \sum_{i=1}^N w_i S^i$, we have the usual spot volatility
\begin{equation*}
    \sigma_I = \sqrt{\sum_{i, j=1}^N w_i w_j \rho_{i j} \sigma_i  \sigma_j}.
\end{equation*}
For the implied variance skew we leave it to the reader to substitute $  f_{i0} =\sigma_i$ and ${f'_{i0}}  = - \eta_i \sigma_i$ into the formula in Equation \eqref{Eq_SpotVarSkew}.
Of course, $\mathcal{S}_I / (2 \sigma_I)$ then gives the implied volatility
skew (at unit time, in the small noise limit). 
\end{example}

\section{Proof of Theorem \ref{Thm_LambExp}}
To emphasize the generality of the argument, we write $\bigl(H, \langle . , . \rangle_H\bigr)$ instead of $\mathcal{H}^N$ with the Cameron-Martin inner product.
The abstract minimization then concerns $ \langle h , \rho^{-1} h \rangle_H/2$ for some invertible (linear) operator $\rho: H \to H$ subject to a constraint $\phi(h)=x$, where $\phi: H \rightarrow \mathbb{R}$ is weakly continuous and thrice Fr\'echet differentiable. The optimization problem can be written as
    \begin{align*}
        \Lambda(x)=\inf \biggl\{ \frac1{2}{\langle h , \rho^{-1} h \rangle_H}:\phi(h)=x   \biggr\}.
    \end{align*}

 \begin{proof}[Proof of Theorem \ref{Thm_LambExp}] 
    Define $\psi:\mathbb{R} \times \mathbb{R} \times H\rightarrow \mathbb{R}\times H$ via 
    \begin{equation*}
        \psi(x,\lambda,h)=\Bigl(\phi(h)-x,\rho^{-1}h-\lambda D\phi(h)\Bigr).
    \end{equation*}
    We want $h^x,\lambda^x$ s.t.\ $\psi(x,\lambda^x, h^x)=(0,0)$, which corresponds to the first order condition of the minimization problem. Regularity of $x \mapsto h^x$ implies regularity of $x \mapsto \Lambda(x)$ in which case we know 
    $\Lambda' = \lambda$. 
    By the implicit function theorem
    \begin{align*}
        \begin{pmatrix}
            \lambda'(x) 
            \\
            h'(x)
        \end{pmatrix}
        =
         -\underbrace{\begin{pmatrix}
            \partial_\lambda \psi_1 & \partial_\lambda \psi_2 \\
            \partial_h \psi_1 & \partial_h \psi_2
        \end{pmatrix}^{-1}}_{=:J^{-1}}
        \begin{pmatrix}
           \partial_x \psi_1 \\
           \partial_x \psi_2
        \end{pmatrix}
    \end{align*}
    A simple calculation tells us, evaluated at $h^x$,
    \begin{align*}
        J=\begin{pmatrix}
            0 & -D\phi(h^x)
            \\
          D\phi(h^x) & \rho^{-1}-\lambda^x\Bigl(...\Bigr).
        \end{pmatrix}
    \end{align*}
    Here the bracket term has no contribution because $\lambda^x=0$ for $x=0$.
    Note that
    \begin{equation*}
        \begin{pmatrix}
           \partial_x \psi_1 \\
           \partial_x \psi_2
        \end{pmatrix}=\begin{pmatrix}
           -1 \\ 0
        \end{pmatrix}\in \mathbb{R}\times H.
    \end{equation*}
    Therefore we only care about the first column of $J$.
    Note that for a block matrix it follows that
    \begin{align*}
        \begin{pmatrix}
            0 & B \\
            C & D
        \end{pmatrix}^{-1}=
        \begin{pmatrix}
            -\Bigl( BD^{-1}C\Bigr)^{-1} & \Bigl( BD^{-1}C\Bigr)^{-1} B D^{-1} \\
             D^{-1}C\Bigl( BD^{-1}C\Bigr)^{-1} & ...
        \end{pmatrix}^{-1}
    \end{align*}
    
    Using the block form of $J$ we see that
    \begin{align*}
        J^{-1}=\begin{pmatrix}
           \Bigl(D\phi(h^x)\, \rho \,D\phi(h^x)\Bigr)^{-1} &  \Bigl(D\phi(h^x)\, \rho \,D\phi(h^x)\Bigr)^{-1} \rho D\phi(h^x)
            \\
          -\Bigl(D\phi(h^x)\, \rho \,D\phi(h^x)\Bigr)^{-1} \rho D\phi(h^x) & ...
        \end{pmatrix}
    \end{align*}
    implying that
 \begin{align}\label{Eq_LbdaPrm}
     \begin{pmatrix}
            \lambda'(0) 
            \\
            h'(0)
        \end{pmatrix}
        =\begin{pmatrix}
            \frac1{\langle \phi'_0,\rho \phi'_0\rangle_H }
            \\
            -\frac1{\langle \phi'_0,\rho \phi'_0\rangle_H } \rho \phi'_0
        \end{pmatrix}.
 \end{align}
For the second derivative we start with a short calculation. Let $g(x)=(\lambda^x,h^x)$ such that $\psi(x,g(x))=0$. Then
\begin{align*}
    0&=\partial_x^2 \psi(x,g(x))=\partial_x\partial_1 \psi(x,g(x))+\partial_x \Bigl(\partial_2 \psi(x,g(x)) g'(x)\Bigr)
    \\
    &=\partial_1^2 \psi(x,g(x)) + 2 \partial_1 \partial_2 \psi(x,g(x)) g'(x) + \partial_2^2 \psi(x,g(x)) \partial_x g(x) \partial_x g(x)
    \\
    &\qquad+\partial_2 \psi(x,g(x)) \partial_x^2g(x).
\end{align*}
Note that 
\begin{align*}
    \partial_1^2 \psi(0,g(0)) =0= 2 \partial_1 \partial_2 \psi(0,g(0)) g'(0),
\end{align*}
and $\partial_2 \psi(x,g(x))=J$. We therefore only need to calculate the $\mathbb{R}^2\times \bigl(\mathbb{R}\times H\bigr)^*\times \bigl(\mathbb{R}\times H\bigr)^*$-tensor $\partial_2^2 \psi(0,g(0))$, which we do component wise.

For this we see $\partial^2_\lambda \psi_1=\partial_\lambda \partial_h \psi_1=\partial^2_\lambda \psi_2=0$. Also $\partial_h^2 \psi_2=0$ at $x=0$. At last $\partial^2_h \psi_1=D^2\phi(h^x)(.,.)$ and $\partial_\lambda \partial_h \psi_2 = -D^2 \phi(h^x)(.,.).$

Evaluationg the tensor $\partial_2^2 \psi(0,g(0))$ at $\partial_x g(x)$ as well as $\partial_x g(x)$  and using that $\partial_2 \psi(x,g(x))=J$ we see
\begin{align*}
    \begin{pmatrix}
            \lambda''(x) 
            \\
            h''(x)
        \end{pmatrix}
        =
         -{J^{-1}}
        \begin{pmatrix}
           D^2\phi(h^x)(h'(x),h'(x)) \\
            -2D^2\phi(h^x)(h'(x),.)\lambda'(x)
        \end{pmatrix}
\end{align*}
Finally, by substituting back we get by \eqref{Eq_LbdaPrm} that
\begin{equation*}
    \lambda''(0)=\frac{-3}{\langle \phi'_0,\rho \phi'_0\rangle_H }\phi''_0\Bigl(\frac1{\langle \phi'_0,\rho \phi'_0\rangle_H } \rho \phi'_0,\frac1{\langle \phi'_0,\rho \phi'_0\rangle_H } \rho \phi'_0\Bigr).
\end{equation*}
\end{proof}

\section{Extension to Stochastic Volatility}

We return to the setting where each component of the index has rough volatility dynamics, as specified in \eqref{equ:2factorRV}, with $2$ Brownian factors. 
An index with $N$-components thus involves a total of $2N$ Brownians, which we assume given as $2N$-dimensional Brownian motion $W$ with non-singular correlation matrix 
\begin{equation*}
\mathbb{E} (W_1 \otimes W_1) = \rho \in \mathbb{R}^{2N\times 2N}.
\end{equation*}
As before, in a small noise regime, the precise form of the model \eqref{equ:2factorRV} is not so important - what matters are the rate functions. To this end, we have
\begin{proposition}[Correlated Index energy]
    Assume that $\sum_{i = 1}^N w_i = 1$ and let $\phi_i :  \cH  \rightarrow
  \mathbb{R}$ be $C^k$-reasonable according to Definition \ref{Def_Ckrsbl}. Let
  \begin{equation*}
      \exp\bigl(\phi (h^1, \ldots, h^{2N})\bigr) \coloneqq \sum_{i = 1}^N w_i \exp \bigl(\phi_i (h^i,h^{N+i})\bigr).
  \end{equation*}
  Then $\phi$ is also reasonable.  Set $\phi_{i0}' = D \phi_i (0) \in \cH^2$ as well as $\phi_{i0}'' = D^2 \phi_i (0) \in (\cH^2 \times \cH^2)^{\star}$ 
  and assume 
  \begin{equation}\label{Eq_SpeedAss2}
      \quad \phi_{i 0}' = 
      \begin{pmatrix}
      \sigma_i
      \\
      0
      \end{pmatrix} \tmop{Id} \in \cH^2, \qquad \tmop{Id}: t \mapsto t,
  \end{equation}
  Then
  \begin{equation*}
      (\sigma^{\varepsilon}_{\tmop{impl}})^2 (x, 1) \sim_{x \downarrow 0} \sigma_I^2
     + x\mathcal{S}_I + o (x),
  \end{equation*}
  with 
  \begin{align*}
      \sigma_I^2 & =  \sum_{i, j=1}^N w_i w_j \rho_{i j} \sigma_i \sigma_j,\\
       \mathcal{S}_I&= - \sigma_I^2 + \sum_{i=1}^N w_i \biggl( \biggl( \sum_{j,\ell=1}^Nw_jw_\ell \sigma_j \sigma_\ell P_{i\ell}\Phi_i P_{ij}\biggr) + \bigl(\sigma_i
    \Sigma_i\bigr)^2 \biggr) \bigg/ \sigma_I^2 ,
    \\
    \qquad \Sigma_i &\assign \sum_{j=1}^N w_j \rho_{i j} \sigma_j, \qquad  P_{i\ell}=\begin{pmatrix}
          \rho_{i\ell}
          \\
          \rho_{(i+N)\ell}
      \end{pmatrix},\\ \Phi_i &=\begin{pmatrix}
          \phi_{i0}''\bigl((\tmop{Id},0);(\tmop{Id},0) \bigr) & \phi_{i0}''\bigl((\tmop{Id},0);(0,\tmop{Id}) \bigr)
          \\
          \phi_{i0}''\bigl((\tmop{Id},0);(0,\tmop{Id}) \bigr) & \phi_{i0}''\bigl((0,\tmop{Id});(0,\tmop{Id}) \bigr)
          \end{pmatrix}. 
  \end{align*}
\end{proposition}

\begin{proof}
We only sketch the proof, because all calculations are similar to the one in the proof of Proposition \ref{Prop_IndEgy}. Similar to before, by power matching we get
  \begin{equation*}
      \langle \phi_0', k \rangle_{\cH^N} = \sum_{i = 1}^N w_i \langle \phi_{i 0}', (k^i,k^{N+i}) \rangle_{\cH^2},
  \end{equation*}
  implying that 
  \begin{equation*}
       \phi_0' = \Bigl(w_1 \bigl(\phi_{10}')^1,
     \ldots, w_N \bigl(\phi_{N0}')^1;w_1 \bigl(\phi_{10}')^2,
     \ldots, w_N \bigl(\phi_{N0}')^2\Bigr) \in \cH^{2N}.
  \end{equation*}
  For the second order similar calculations show
  \begin{equation*}
      \phi_0'' (k, k) = - \langle \phi_0', k \rangle^2_{\cH^{2N}} + \sum_{i = 1}^N
     w_i \Bigl(\phi_{i 0}'' \Bigl(k^i, k^{N+i};k^i,k^{N+i}\Bigr) + \langle \phi_{i 0}', (k^i,k^{N+i}) \rangle_{\cH^2}^2\Bigr).
  \end{equation*}
    We note that
  \begin{equation*}
    \bigl(\rho \phi_0'\bigr)^i = 
    \sum_{j=1}^{2N}\rho_{i j} \bigl(\phi_{0}'\bigr)^j
    =\sum_{j=1}^N \rho_{i j} w_j \bigl(\phi_{j 0}'\bigr)^1+\sum_{j=1}^N \rho_{i (j+N)} w_j \bigl(\phi_{j 0}'\bigr)^2
   .
  \end{equation*}
    Under the assumption of \eqref{Eq_SpeedAss2} we thus see
    \begin{equation*}
        \sigma_I^2=\langle \rho \phi_0',\phi_0'\rangle_{\cH^N}=\sum_{i,j=1}^N w_iw_j\rho_{ij} \sigma_i\sigma_j.
    \end{equation*}
    We split up $\mathcal{S}_I = \phi_0'' (\rho   \phi_0',
  \rho   \phi_0') / \sigma_I^2$  into
  \begin{equation*}
       \mathcal{S}_I = (\mathcal{S}_1 +\mathcal{S}_2 +\mathcal{S}_3) /
     \sigma_I^2.
  \end{equation*}
  Under assumption \eqref{Eq_SpeedAss2} a calculation shows that
    \begin{align*}
        \mathcal{S}_1=- \biggl( \sum_{i=1}^N  \sum_{j=1}^N w_i w_j  \rho_{ij}\sigma_i \sigma_j
       \biggr)^2 =-\sigma_I^4
    \end{align*}
     Doing similar calculations as in Proposition \ref{Prop_IndEgy}, under the assumption of \eqref{Eq_SpeedAss2} it follows that
    \begin{align*}
      \mathcal{S}_2 &= \sum_{i = 1}^N w_i 
      \biggl( \sum_{j,\ell=1}^Nw_jw_\ell \sigma_j \sigma_\ell P_{i\ell}\Phi_i P_{ij}\biggr) .
  \end{align*}
  as well as
  \begin{equation*}
      \mathcal{S}_3  = \sum_{i = 1}^N w_i
     \biggl(\sigma_i  \sum_{j=1}^N w_j \rho_{i j} \sigma_j \biggr)^2
   = \sum_{i = 1}^N w_i  \bigl(\sigma_i
    \Sigma_i\bigr)^2.
    \end{equation*}
    The formula for $\cS_I$ then follows from summing up.
\end{proof}

\begin{remark}
As before, the proof shows that one can do without assumption \eqref{Eq_SpeedAss2} at the expense of less appealing formulae. Yet, 
this assumption is satisfied in the examples we have in mind. Indeed, let us show that equation \eqref{Eq_SpeedAss2} is satisfied in the rough volatility case, i.e.\ when $\phi_i$ is given by\footnote{No need for a correlation parameter here, contained in Hilbert structure of $\mathcal{H}^2 \ni (h^1,h^2)$.}
    \begin{gather*}
        \phi_i(h^1,h^2)=\int_0^1 f_i(\hat h^2) \,dh^1.
    \end{gather*}
    In that case 
    \begin{align*}
        D\phi_i(h^1,h^2)(g^1,g^2)=\int_0^1 f_i(\hat h^2) \,dg^1 + \int_0^1 f_i'(\hat h^2)\hat g^2 \,dh^1.
    \end{align*}
    This implies that $\phi_{i0}'$ satisfies Assumption \eqref{Eq_SpeedAss2} with $\sigma_i=f_i(0)$.
    \end{remark}

\section{Technical Results} 
In the main text we have used the classical Cameron--Martin Hilbert space.  Taking a general separable Hilbert space $H$ instead, we
here consider, in this Hilbert generality,
 \[ \Lambda (x) = \inf \left\{ \frac{1}{2} \langle h, h \rangle
     _H : \phi (h) = x \right\} \in [0, \infty] . 
 \]
Call \emph{$x$-admissible} any $h \in H$ with $\phi (h ) = x$. Set $D _{\Lambda} = \{ x \in \mathbb{R}: \Lambda (x) < \infty \}$.
                 
\subsection{Monotonicity of Energy} 

We now discuss monotonicity of $\Lambda$.

\begin{proposition}\label{Prop_MonRF}
  Assume $\phi: H \to \mathbb{R}$ is continuous and $\phi (0) = 0$. Let $0 < x < y.$ Then $0 = \Lambda (0)
  \leqslant \Lambda (x) \leqslant \Lambda (y) .$
\end{proposition}

\begin{proof}
  If $\Lambda (y) = + \infty$ there is nothing to show, else have, for every
  $\varepsilon > 0$, some $y$-admissible $h^y \in H$:  $\frac{1}{2} \langle h^y, h^y \rangle_{_H} < \Lambda (y) + \varepsilon$. The real function $[0, 1]
  \ni c \mapsto \phi (c h^y) \in \mathbb{R}$ is continuous, by continuity of
  $\phi$, with end-points $\phi (0) = 0$ and $\phi \left( {h^y}  \right) = y$,
  respectively. By the intermediate value theorem, for every $x \in (0, y)$,
  there is $c_0 \in (0, 1)$ such that $\phi (c_0 h^y) = x$ and so
  \begin{equation*}\Lambda (x) \leqslant \frac{1}{2} \langle c_0 h^y, c_0 h^y
     \rangle_H < \frac{1}{2} \langle h^y, h^y \rangle
     _{H} < \Lambda (y) + \varepsilon . 
     \end{equation*}
  Conclude by taking $\varepsilon \downarrow 0$. \ 
\end{proof}

\subsection{Existence of a Minimizer} 

\begin{proposition}\label{Prop_ExMin}
  Assume $\phi: H \to \mathbb{R}$ is weakly continuous (hence continuous).
  \renewcommand{\labelenumi}{(\roman{enumi})}
  \begin{enumerate}
      \item Let $x \in \phi (H)$. Then there exists an $x$-admissible $h^x$ s.t.\
  $\Lambda (x) = \frac{1}{2} \langle h^x, h^x \rangle_{H}$.
  \item  The map $\Lambda$ is LSC.
  \end{enumerate}
\end{proposition}

\begin{proof} We only consider $x>0$.
  (i) Fix $x \in \phi (H)$, so that $\Lambda (x) < \infty$ and pick $x$-admissible
  $h^n \in H$ so that $\frac{1}{2} \langle h^n, h^n \rangle _{_H}
  \downarrow \Lambda (x)$. By weak compactness, there exists $h \in H$ and a subsequence $(n_k)$
  such that $h^{{n_k} } \rightarrow h$ weakly in $H$. Hence $x = \phi (h^{n })
  \rightarrow \phi (h)$ which shows that $h$ is $x$-admissible. If follows
  that
  \[ \Lambda (x) \leqslant \frac{1}{2} \| h \|^2_H \leqslant \liminf_{k
     \rightarrow \infty} \frac{1}{2} \| h^{n_k} \|_H^2 = \Lambda (x) . \]
     (ii) Consider $x_n \to x$. We need to see $\Lambda (x) \le \liminf_n \Lambda(x_n)$. We may assume that $\liminf_n \Lambda(x_n)<\infty$ and that all $x_n \in D_\Lambda$ as otherwise it is inconsequential to remove all $n$ with
     $\Lambda(x_n) = \infty$. Consider $(h^{x_n})_n$, then there is a subsequence $(n_k)$ such that $(h^{x_{n_k}})$ is bounded in $H$ and therefore weakly compact. By weak continuity any weak limit point $h$ is $x$-admissible, implying that $\Lambda (x) \le \frac{1}{2} \| h \|^2_H$. We conclude with lower semi-continuity of the norm in $H$ under weak convergence. 
\end{proof}

We give a general criterion for weak continuity. 

\subsection{Rough Path type Continuity implies Weak Continuity}  

Consider a Banach space $\mathbf{{W}} = W \oplus \tilde{W} $ which is a
direct sum of two Banach spaces $W, \tilde{W}$. Write $\pi :
\mathbf{{W}} \rightarrow W$ for the canonical projection on the
first component. Let $H$ be a Hilbert space, with compact embedding $\iota : H
\hookrightarrow W$, and further consider a lift, that is a map
\[ \mathcal{L}: H \rightarrow \mathbf{{W}} \]
so that $\pi  \circ \mathcal{L}= \iota$. \ Assume that the function
   \begin{equation*}
       I (\mathbf{w}) =\begin{cases}
           \tfrac{1}{2} \langle h, h \rangle_H &\text{when }
   \mathbf{w} =\mathcal{L} (h), \\
   +\infty &\tmop{otherwise}.
       \end{cases}
   \end{equation*}
is a good rate function on $\mathbf{{W}}$, i.e. LSC with compact
level sets. This situation is typical for small noise large deviations of Gaussian rough paths.

\begin{theorem} Assume $I$ is a good rate function, then $\mathcal{L}$ is weakly continuous. As a consequence, any map
  \[ \phi : H \rightarrow \mathbb{R}, \quad \phi = \bar{\phi} \circ
     \mathcal{L}, \quad \bar{\phi} \in C (\mathbf{{W}}, \mathbb{R})
  \]
  is also weakly continuous.
\end{theorem}

\begin{remark} (i) 
  We do not (want to) assume that $\phi = \hat{\phi} \circ \iota$ for $ \hat{\phi} \in C (\mathcal{W}, \mathbb{R})$
 \noindent (ii) Consider
  \[ \Lambda (x) = \inf \left\{ I( \mathbf{w} ) : \bar \phi ( \mathbf{w} ) = x \right\} \in [0, \infty].
 \]
 The contraction principle then tells us that $\Lambda$ is LSC, has compact levels sets, and also gives the existence of $h^x$.
 Our presentation highlights the role of weak continuity, which may (or may not) be checked with rough path type continuity arguments.
\end{remark}

\begin{proof}
  Consider $h^n \rightarrow h$ weakly in $H$. Such a sequence is necessarily
  bounded, as a consequence of the uniform boundedness principle. By goodness of the rate function,
  $\mathcal{L} (h^{{n_k} }) = \mathbf{w }^k \rightarrow \mathbf{w}$ in
  $\mathbf{{W}}$, for some $\mathbf{w} = (w, \tilde{w}) \in
  \mathbf{{W}}$ and some subsequence $(n_k)$. LSC implies
  \begin{equation*}
      I (\mathbf{w}) \leqslant \liminf_{k \rightarrow \infty} I (\mathbf{w
     }^k) \leqslant \sup_k \tfrac{1}{2} \langle h^k, h^k \rangle_H < \infty.
  \end{equation*}
  It follows that $\mathbf{w} =\mathcal{L} (w)$ with $w \in H.$ To identify
  $w$, note $h^{n_k} \rightarrow w$ in $\mathcal{W}$ and since
  $\mathcal{W}^{\star} \subset H^{\star} \cong H$, we see $\langle \theta,
  h^{n_k} \rangle_H \rightarrow \langle \theta, w \rangle_H$ for $\theta \in \cW$. But by weak
  convergence we also have $\langle \theta, h^{n_k} \rangle_\cW \rightarrow
  \langle \theta, h \rangle_\cW = \langle \theta, h \rangle_H $. This implies
  $w = h$. We have shown that $h^n \rightarrow h$ weakly implies $\mathcal{L}
  (h^{n }) \rightarrow \mathcal{L} (h)$ in $\mathbf{{W}}$ along a
  subsequence. By a standard argument this also shows convergence (without
  subsequence).
\end{proof} 

\subsection{Continuity of Rate Function}

We start with an explicit example where one has discontinuities.

\begin{example}\label{Ex_ContRF} 
  Assume $\phi (h) = F (\langle g, h \rangle_H)$ for some continuous $F$
  with $F (0) = 0$ and a fixed unit vector $g \in H$. By scaling, this
  applies to any non-zero $g \in H$ and the case $g = 0$ is trivial anyway.
  Such $\phi$ is obviously weakly continuous on $H$.
  Assume $0 \leqslant x \in \phi (H)$ and $F$ strictly increasing, then $\phi
  (h) = x$ iff $\langle g, h \rangle_H = F^{- 1} (x)$. Obviously the
  minimal $h$ is colinear to $g$ and explicitly $h = F^{- 1} (x) \frac{}{}
  g$. Then
  \[ \Lambda (x) = \frac{1}{2} (F^{- 1} (x))^2, \]
  which is in fact continuous since $F^{- 1}$ is. 
  If $F$ is only assumed to
  be increasing (in sense of non-decreasingness), we write $F^-$ for the
  generalized-inverse of $F$ which is defined on the interval $F (\mathbb{R})$ by
  \[ F^- (y) = \inf \{ t \in \mathbb{R}: F (t) = y \}, \quad y \in F
     (\mathbb{R}) . \]
  Such $F^-$ is also increasing and left-continuous, hence LSC. Flat parts of
  $F$ precisely correspond to jumps in $F^-$. As above, $h^x = F^- (x) g$ and
  hence
  \[ \Lambda (x) = \frac{1}{2} \bigl(F^- (x)\bigr)^2 . \]
  This function need not be continuous because $F^-$ may have jumps.
  
  If $F$ is not increasing, this form of the rate function \emph{in
  general} fails and one can only say
  \[ \Lambda (x) = \min \left\{ \frac{1}{2} | y |^2 : F (y) = x \right\} . \]
\end{example}

\begin{proposition}\label{Prop_ContRF} Assume $x \in D _{\Lambda}$ admits an $x$-admissible minimizer $h^x \in H$.
Assume $\phi: H \to \mathbb{R}$ has no local maximum at $h^x$. Then $\Lambda$ is continuous at $x$.
This holds in particular, if $\phi$ is $C^1$-reasonable, see Definition \ref{Def_Ckrsbl}.
\end{proposition}

\begin{proof}
  We already know that $\Lambda$ is monotone und LSC, implying that $\Lambda$
  is left-continuous. Thus the only possible disconuity can be a jump at some
  point $x \in D _{\Lambda}$.
  
  Let $x \in D_{\Lambda}, x > 0$ and $\varepsilon > 0$ be arbirtray. By
  Assumption there is some $x$-admissible $h^x \in {H}$ minimizing $\Lambda (x)$. Choose $\delta$ so small that 
  \begin{equation*}
      \Bigl|\frac{1}{2} \langle h, h \rangle_H - \frac{1}{2} \langle h^x, h^x \rangle_H
  \Bigr| < \frac{\varepsilon}{2}
  \end{equation*}
  for all $h \in U_{\delta} (h^x)$. By
  assumption $h^x$ is not a local maxima therefore there
  is some $\tilde{h} \in U_{\delta} (h^x)$ such that $x = \phi (h^x)
  < \phi (\tilde{h}) \backassign \tilde{y} .$ But by monotonicity of $\Lambda$
  and by construction
  \[ \Lambda (x) < \Lambda (\tilde{y}) \le \frac{1}{2} \langle \tilde{h},
     \tilde{h} \rangle \le \frac{1}{2} \langle h^x, h^x \rangle +
     \frac{\varepsilon}{2} = \Lambda (x) + \frac{\varepsilon}{2} \]
  As $\varepsilon$ was arbitrary we see that ${\lim_{\lambda \rightarrow 0}} 
  \Lambda (x + \lambda) = \Lambda (x)$ implying that $\Lambda$ is continuous
  at point $x$ and therefore in $D_{\Lambda}$. \ 
\end{proof}

Even if the energy function is continuous, it need not be smooth. Similar facts for 
(sub)Riemannian square-distance are well-known. In the context of Example \ref{Ex_ContRF}
we can exhibit this directly via the function $F(y) = |y|^{\alpha}, \alpha > 0$ with inverse $\pm |y|^{1/\alpha}$.
Then 
\[ \Lambda (x) = \frac{1}{2} \bigl(F^{- 1} (x)\bigr)^2 = \frac{1}{2} |x|^{2/\alpha}. \]

Note that $\phi: h \mapsto F ( \langle g, h \rangle )$ is weakly continuous and inherits Fr\'echet regularity from $F$.
For instance, $F(y) = y^2$ makes $\phi$ Fr\'echet smooth, but $\Lambda$ fails to be $C^1$ at $x=0$; the problem, as seem below, 
is that $D\phi(0) = F'(0) g = 0$ in this example.

\subsection{Smoothness of Energy Function} 

The following is a consequence of a more general statement that can be found in the appendix of \cite{friz2021precise}.
\begin{theorem} Assume $\phi:H \to \mathbb{R}$ is $C^n$-reasonable, see Definition \ref{Def_Ckrsbl}.
Then for all sufficiently small $x$, there exists a unique $x$-admissible minimizer $h^x$, such that $x \mapsto h^x \in H$ is $C^{n-1}$. Moreover, $x \mapsto \Lambda (x)$ is  $C^{n}$ near $x=0$, hence
\begin{equation*}
    \Lambda (x) = \Lambda ''(0) x^2/2! + \Lambda '''(0) x^3/3! +\dots +\Lambda^{(n)}(0) x^{n}/n! + o(|x|^{n}).
\end{equation*}

\end{theorem}

\bibliographystyle{alpha}
\bibliography{Bibliography}

\newcommand{\etalchar}[1]{$^{#1}$}
\begin{thebibliography}{ABOBF03}

\bibitem[ABOBF02]{avellaneda2002reconstruction}
Marco Avellaneda, D.~Boyer-Olson, J.~Busca, and P.~Friz.
\newblock Reconstruction of volatility: Pricing index options by the steepest
  descent approximation.
\newblock {\em Risk}, 15:87--91, 2002.

\bibitem[ABOBF03]{avellaneda2003applicaton}
Marco Avellaneda, D.~Boyer-Olson, J.~Busca, and P.~Friz.
\newblock Application of large deviation methods to the pricing of index
  options in finance.
\newblock {\em Comptes Rendus Mathematique}, 336(3):263--266, 2003.

\bibitem[BBF04]{berestycki2004computing}
H.~Berestycki, J.~Busca, and I.~Florent.
\newblock Computing the implied volatility in stochastic volatility models.
\newblock {\em Communications on Pure and Applied Mathematics},
  57(10):1352--1373, 2004.

\bibitem[Ber16]{bergomi2016stochastic}
L.~Bergomi.
\newblock {\em Stochastic Volatility Modeling}.
\newblock A Chapman et Hall book. CRC Press, 2016.

\bibitem[BFF{\etalchar{+}}23]{bayer2023rough}
Christian Bayer, Peter Friz, Masaaki Fukasawa, Jim Gatheral, Antoine Jacquier,
  and Mathieu Rosenbaum.
\newblock {\em Rough volatility}.
\newblock To appear, 2023.

\bibitem[BFG16]{bayer2016pricing}
C.~Bayer, P.~Friz, and J.~Gatheral.
\newblock Pricing under rough volatility.
\newblock {\em Quantitative Finance}, 16(6):887--904, 2016.

\bibitem[BFG{\etalchar{+}}19]{bayer2019shorttime}
C.~Bayer, P.~Friz, A.~Gulisashvili, B.~Horvath, and B.~Stemper.
\newblock Short-time near-the-money skew in rough fractional volatility models.
\newblock {\em Quantitative Finance}, 19(5):779--798, 2019.

\bibitem[BFG{\etalchar{+}}20]{bayer2020regularity}
Christian Bayer, Peter~K Friz, Paul Gassiat, Jorg Martin, and Benjamin Stemper.
\newblock A regularity structure for rough volatility.
\newblock {\em Mathematical Finance}, 30(3):782--832, 2020.

\bibitem[BFL15]{bayer2015pdf}
C.~Bayer, P.~Friz, and P.~Laurence.
\newblock On the probability density function of baskets.
\newblock In {\em Large Deviations and Asymptotic Methods in Finance}, pages
  449--472, Cham, 2015. Springer International Publishing.

\bibitem[BS13]{brunick2013mimicking}
Gerard Brunick and Steven Shreve.
\newblock Mimicking an it{\^o} process by a solution of a stochastic
  differential equation.
\newblock {\em The Annals of Applied Probability}, 23(4):1584--1628, 2013.

\bibitem[CB15]{Cont2015forward}
Rama Cont and Amel Bentata.
\newblock Forward equations for option prices in semimartingale models.
\newblock {\em Finance and Stochastics}, 19, 01 2015.

\bibitem[DFJV14a]{deuschel2014marginal}
Jean-Dominique Deuschel, Peter~K Friz, Antoine Jacquier, and Sean Violante.
\newblock Marginal density expansions for diffusions and stochastic volatility
  i: Theoretical foundations.
\newblock {\em Communications on Pure and Applied Mathematics}, 67(1):40--82,
  2014.

\bibitem[DFJV14b]{deuschel2014marginal2}
Jean-Dominique Deuschel, Peter~K Friz, Antoine Jacquier, and Sean Violante.
\newblock Marginal density expansions for diffusions and stochastic volatility
  i: Theoretical foundations.
\newblock {\em Communications on Pure and Applied Mathematics}, 67(1):40--82,
  2014.

\bibitem[DK94]{derman1994riding}
Emanuel Derman and Iraj Kani.
\newblock Riding on a smile.
\newblock {\em Risk}, 7(2):32--39, 1994.

\bibitem[DMF18]{de2018local}
Stefano De~Marco and Peter~K Friz.
\newblock Local volatility, conditioned diffusions, and {V}aradhan's formula.
\newblock {\em SIAM Journal on Financial Mathematics}, 9(2):835--874, 2018.

\bibitem[DMFG13]{de2013rational}
Stefano De~Marco, Peter Friz, and Stefan Gerhold.
\newblock Rational shapes of local volatility.
\newblock {\em Risk}, 26(2):70, 2013.

\bibitem[Dup94]{dupire1994pricing}
Bruno Dupire.
\newblock Pricing with a smile.
\newblock {\em Risk}, 7(1):18--20, 1994.

\bibitem[FGG{\etalchar{+}}15]{friz2015large}
Peter~K Friz, Jim Gatheral, Archil Gulisashvili, Antoine Jacquier, and Josef
  Teichmann.
\newblock {\em Large deviations and asymptotic methods in finance}.
\newblock Series: Springer Proceedings in Mathematics \& Statistics, Vol. 110,
  2015, 2015.

\bibitem[FGP21]{friz2021precise}
P.~Friz, P.~Gassiat, and P.~Pigato.
\newblock Precise asymptotics: Robust stochastic volatility models.
\newblock {\em The Annals of Applied Probability}, 31(2):896 -- 940, 2021.

\bibitem[FPS21]{friz2021step}
Peter Friz, Paolo Pigato, and Jonathan Seibel.
\newblock The step stochastic volatility model.
\newblock {\em Risk}, June 2021.

\bibitem[FT22]{fukasawa2022partial}
Masaaki Fukasawa and Ryoji Takano.
\newblock A partial rough path space for rough volatility.
\newblock {\em arXiv preprint arXiv:2205.09958}, 2022.

\bibitem[FZ17]{forde2017asymptotics}
M.~Forde and H.~Zhang.
\newblock Asymptotics for rough stochastic volatility models.
\newblock {\em SIAM Journal on Financial Mathematics}, 8(1):114--145, 2017.

\bibitem[Gat06]{gatheral2006volatility}
Jim Gatheral.
\newblock {\em The volatility surface: a practitioner's guide}.
\newblock John Wiley \& Sons, 2006.

\bibitem[GHL{\etalchar{+}}12]{gatheral2012asymptotics}
J.~Gatheral, E.~Hsu, P.~Laurence, C.~Ouyang, and T.-H. Wang.
\newblock Asymptotics of implied volatility in local volatility models.
\newblock {\em Mathematical Finance}, 22(4):591--620, 2012.

\bibitem[GHL14]{guyon2014nonlinear}
J.~Guyon and P.~Henry-Labord\`ere.
\newblock {\em Nonlinear option pricing}.
\newblock Chapman \& Hall/CRC Financial Mathematics Series. CRC Press, Boca
  Raton, FL, 2014.

\bibitem[Gul18]{Gulisash2017}
Archil Gulisashvili.
\newblock Large deviation principle for volterra type fractional stochastic
  volatility models.
\newblock {\em SIAM Journal on Financial Mathematics}, 9(3):1102--1136, 2018.

\bibitem[Gul22]{gulisashvili2022multivariate}
Archil Gulisashvili.
\newblock Multivariate stochastic volatility models and large deviation
  principles.
\newblock {\em arXiv preprint arXiv:2203.09015}, 2022.

\bibitem[Gy{\"o}86]{gyongy1986mimicking}
Istv{\'a}n Gy{\"o}ngy.
\newblock Mimicking the one-dimensional marginal distributions of processes
  having an it{\^o} differential.
\newblock {\em Probability theory and related fields}, 71(4):501--516, 1986.

\bibitem[HL08]{henry2008analysis}
P.~Henry-Labord\`ere.
\newblock {\em Analysis, geometry, and modeling in finance: Advanced methods in
  option pricing}.
\newblock Chapman and Hall/CRC, 2008.

\bibitem[HR98]{hobson1998models}
D.~Hobson and L.~C.~G. Rogers.
\newblock Complete models with stochastic volatility.
\newblock {\em Math. Finance}, 8(1):27--48, 1998.

\bibitem[JP22]{jacquier2022large}
Antoine Jacquier and Alexandre Pannier.
\newblock Large and moderate deviations for stochastic volterra systems.
\newblock {\em Stochastic Processes and their Applications}, 149:142--187,
  2022.

\bibitem[Osa15]{osajima2015general}
Y.~Osajima.
\newblock General asymptotics of {W}iener functionals and application to
  implied volatilities.
\newblock In {\em Large deviations and asymptotic methods in finance}, pages
  137--173. Springer, 2015.

\bibitem[Var67]{varadhan1967diffusion}
Srinivasa~RS Varadhan.
\newblock Diffusion processes in a small time interval.
\newblock {\em Communications on Pure and Applied Mathematics}, 20(4):659--685,
  1967.

\end{thebibliography}

\end{document}